\documentclass[journal,final]{IEEEtran}%
\usepackage{amssymb}
\usepackage{amsmath}
\usepackage{amsfonts}
\usepackage{cite}%
\usepackage{graphics}
\usepackage{epsfig}
\usepackage{hyperref}
\providecommand{\U}[1]{\protect\rule{.1in}{.1in}}
\newtheorem{theorem}{Theorem}

\newtheorem{condition}{Assumption}

\newtheorem{definition}{Definition}

\newtheorem{lemma}{Lemma}

\newtheorem{remark}{Remark}

\begin{document}

\title{Effect of Coupling on the Epidemic Threshold in Interconnected Complex
Networks: A Spectral Analysis}
\author{Faryad Darabi Sahneh, Caterina Scoglio, Fahmida N. Chowdhury \thanks{F. D.
Sahneh \& C. Scoglio are with the Department of electrical and Computer
Engineering at Kansas State University, Manhattan, KS, 66502. -E-mail:
\{faryad,caterina\}@ksu.edu}\thanks{F. N. Chowdhury is with the Directorate
for Social, Behavioral \& Economic Sciences, National Science Foundation,
Arlington, VA, USA-E-mail: fchowdhu@nsf.gov}}
\maketitle

\begin{abstract}
In epidemic modeling, the term infection strength indicates the ratio of
infection rate and cure rate. If the infection strength is higher than a
certain threshold -- which we define as the epidemic threshold - then the
epidemic spreads through the population and persists in the long run. For a
single generic graph representing the contact network of the population under
consideration, the epidemic threshold turns out to be equal to the inverse of
the spectral radius of the contact graph. However, in a real world scenario it
is not possible to isolate a population completely: there is always some
interconnection with another network, which partially overlaps with the
contact network. Results for epidemic threshold in interconnected networks are
limited to homogeneous mixing populations and degree distribution arguments.
In this paper, we adopt a spectral approach. We show how the epidemic
threshold in a given network changes as a result of being coupled with another
network with fixed infection strength. In our model, the contact network and
the interconnections are generic. Using bifurcation theory and algebraic graph
theory, we rigorously derive the epidemic threshold in interconnected
networks. These results have implications for the broad field of epidemic
modeling and control. Our analytical results are supported by numerical simulations.

\end{abstract}

\section{Introduction}

In the existing individual-based epidemic models, the interaction and
consequently the infection spreading process is driven by a single graph, the
contact network by which individuals are in physical contact. However, in
order to study epidemics in cyber-physical systems, a more elaborate
description of the interaction is required. Several researchers from computer
science, communication, networking, and control communities are working on
describing this complex interaction by using multiple interconnected networks
\cite{Havlin2011PRE, Brummitt2012PNAS}. The study of the spreading of
epidemics in interconnected networks is a major challenge of complex networks,
which has recently attracted substantial attention \cite{funk2010PRE,
Dickison2012PRE, Mendiola12PRE, Wang2011wicom}.

The N-Intertwined Mean-Field Approximated (NIMFA) model, first proposed by Van
Mieghem \cite{Piet2009TN}, pointed out the specific role of a general network
on the spreading process. Before NIMFA, most network-based epidemic models
considered aggregated networks characterized by a given node degree
distribution \cite{Vespignani2001PRE}. The NIMFA model has triggered a
pervasive amount of research on epidemic spreading on general networks, in
different scenarios and with different compartments \cite{Dorogovtsev2012PRL,
FaryadCDC11SAIS, FaryadCDC12IDN, Mina2011JTB, Ferreira2012arXiv,
preciado2010arXiv}. The key aspect of this class of models relay on the use of
rigorous spectral analysis to determine the evolution of the epidemic. In
particular, it is found that if the infection strength is higher than a
certain threshold -- which is defined as the epidemic threshold
\cite{Piet2009TN}, then the epidemic spreads through the population and
persists in the long run. For a single generic graph representing the contact
network of the population under consideration, the epidemic threshold turns
out to be equal to the inverse of the spectral radius of the contact graph.
The epidemic threshold provides a measure of the network robustness with
respect to epidemics: the larger the epidemic threshold is, the more robust
the network is, since more instances of infection strength will not spread in
the long run.

Current research efforts are directed to establish the impact of
interconnected networks in spreading processes. Interconnection of networks
can only make the system more vulnerable to infection propagation. Therefore,
it is expected that the epidemic threshold of a network is not increased when
it is connected to another network. Results for epidemic threshold in
interconnected networks are limited to homogeneous mixing populations and
degree distribution arguments. In \cite{Dickison2012PRE}, two networks
following the standard configuration model and interconnected with their own
intranetwork are studied. Two possible scenarios are considered:
strongly-coupled networks and weakly coupled networks. In strongly-coupled
epidemics, either the epidemic invades both networks or not spread at all. In
contrast, in weakly-coupled network systems, an intermediate scenario can
happen where an epidemic spreads in one network but does not invade the
coupled network.

The objective of this paper is to study how much the interconnection can
affect the robustness of the network when considering two general networks
with network topologies expressed by their adjacency matrices. In this paper,
we study the spreading process of a susceptible-infected-susceptible (SIS)
type epidemic model in an interconnected network of two general graphs. First,
we prove that interconnection always increases the probability of infection.
Second, we find that the epidemic threshold for a network interconnected to
another network with a given infection strengths rigorously derived as the
spectral radius of a new matrix which accounts for the two networks and their
interconnection links. We make use of algebraic graph theory to demonstrate
our findings.

The main contribution of this paper is the use of spectral analysis to analyze
epidemic spreading in interconnected networks. To the best of our knowledge,
this is the first time such approach has been used. As a result of our
analysis on two general interconnected networks, we show that assumptions on
the level of connectedness are not necessary to determine the epidemic
threshold, and the evolution of the spreading process. Consequently, our
results are rigorous and general, and reproduce results of
\cite{Dickison2012PRE}, as a specific case.

This rest of the paper is organized as follows. Preliminary tools in graph
theory and a background in epidemic modeling is the subject of Section
\ref{Sec: Preliminary}. In Section \ref{Section: Modeling}, SIS epidemic
spreading in two interconnected network is modeled. Main results on the
epidemic threshold are provided in Section \ref{Section: Mainresult}. Finally,
simulation results are available in Section \ref{Section: Simulation}.

\section{Preliminary and Background\label{Sec: Preliminary}}

\subsection{Individual-Based Epidemic Models}

Epidemic modeling has a rich history. Biological epidemiology has produced
significant number of deterministic and stochastic models. These models have
been successful in providing insights and deep understanding of the epidemic
process phenomenon leading to successful conclusions about prevention and
prediction of epidemics. In \cite{bailey1975Book}, a stochastic epidemic model
was studied for a well-mixed homogenous population. However, this assumption
on the population appeared to be too simplistic in order to capture realistic
cases. The theory of random networks was employed to generate models to
represent contact patterns among individuals within a population.
Specifically, results were reported in \cite{Vespignani2002EPJB} for
heterogeneous networks and in \cite{Vespignani2001PRE} for scale free
networks. In the search for detailed and general models, individual-based
epidemic models were proposed, where the contact network is represented by a
graph. In particular, in the NIMFA model \cite{Piet2009TN}, the probability of
infection for each individual are the system states. For this model, the
epidemic threshold is shown to be equal to the inverse of the spectral radius
of the contact graph.

\subsection{Graph Theory}

Graph theory (see \cite{Diestel97Book, Piet2011graphspectra}\textbf{) }is
widely used for representing the contact topology in an epidemic network. Let
$\mathcal{G}=\left\{  \mathcal{V},\mathcal{E}\right\}  $ represent a directed
graph, and $\mathcal{V=}\left\{  1,...,N\right\}  $\ denote the set of
vertices. Every agent is represented by a vertex. The set of edges is denoted
by $\mathcal{E\subset V\times V}$. An edge is an ordered pair $(i,j)\in
\mathcal{E}$ if agent $i$ can potentially be directly infected by agent $j$.
$\mathcal{N}_{i}=\left\{  j\in\mathcal{V\mid}(i,j)\in\mathcal{E}\right\}  $
denotes the neighborhood set of vertex $i$. Graph $\mathcal{G}$\ is said to be
undirected if for any edge $(i,j)\in\mathcal{E}$, edge $(j,i)\in\mathcal{E}$.
In this paper, we assume that there is no self loop in the graph, i.e.,
$(i,i)\notin\mathcal{E}$, and the contact graph is undirected. A path is
referred by the sequence of its vertices. A path $\mathcal{P}$ of length $k$
between $v_{0}$, $v_{k}$\ is the ordered sequence $(v_{0},...,v_{k})$ where
$(v_{i-1},v_{i})\in\mathcal{E}$ for $i=1,...,k$. Graph $\mathcal{G}$\ is
connected if any two vertices are connected with a path in $\mathcal{G}$.
$\mathcal{A=}\left[  a_{ij}\right]  \in%
\mathbb{R}
^{N\times N}$ denotes the adjacency matrix of $\mathcal{G}$, where $a_{ij}=1$
if and only if $(i,j)\in\mathcal{E}$ else $a_{ij}=0$. A graph is connected iff
its associated adjacency matrix is irreducible. The largest eigenvalue of the
adjacency matrix $A$ is called \emph{spectral radius} of $A$ and is denoted by
$\lambda_{1}(A)$.

A network of two interconnected graphs $\mathcal{G}_{1}$ and $\mathcal{G}_{2}
$ is represented by the set of non-overlapping vertices $\{\mathcal{V}%
_{m}\}_{m=1}^{2}$ and the set of edges $\{\mathcal{E}_{mn}\}_{m,n=1}^{2}$,
where $\mathcal{E}_{mn}\mathcal{\subset V}_{m}\mathcal{\times V}_{n}$ denotes
the connection between vertices of $\mathcal{V}_{m}$ to vertices of
$\mathcal{V}_{n}$ for each $m,n\in\{1,2\}$. Connections between vertices of
$\mathcal{V}_{m}$ to vertices of $\mathcal{V}_{n}$ can be represented by
matrices $A_{mn}$, for each $m,n\in\{1,2\}$.

In order to account for the hops of a path between the two interconnected
graphs, we define the following class of paths. Without loss of generality, we
assume that each path starts from $i\in\mathcal{V}_{1}$.

\begin{definition}
A path from node $i\in\mathcal{G}_{1}$ to node $j$ is of class $(l_{1}%
,...,l_{s})$, with non-negative integers $l_{1},...,l_{s}$, if it first make
$l_{1}$ jumps in $\mathcal{G}_{1}$ then goes to $\mathcal{G}_{2}$ and make
$l_{2}$ jumps in $\mathcal{G}_{2}$ then goes back to $\mathcal{G}_{1}$ and
makes $l_{3}$ jumps in $\mathcal{G}_{1}$ and so on until it makes the last
$l_{s}$ jumps to reach $j$.
\end{definition}

It can be inferred from the above definition that a path of class
$(l_{1},...,l_{s})$, has length $L=(s-1)+$ $l_{1}+\cdots+l_{s}$.

\section{Modeling SIS Spreading in Interconnected
Networks\label{Section: Modeling}}

In this paper, we study the spreading process of a
susceptible-infected-susceptible (SIS) type epidemic model in an
interconnected network of two graphs. In order to develop the model for the
case of interconnected network, first we review the model for SIS\ spreading
over a single graph.

\subsection{SIS Epidemic Spreading over a Single Graph\label{SIS_Single}}

Consider a network of $N$ agents where the contact is determined by the
adjacency matrix $A$. Agent $j$ is a neighbor of $i$, denoted by
$j\in\mathcal{N}_{i}$, if it can contract the infection to agent $i$. If $j$
is a neighbor of $i$ then $a_{ij}=1$, otherwise $a_{ij}=0$. In the SIS\ model,
the state $X_{i}(t)$ of an agent $i$ at time $t$ is a Bernoulli random
variable, where $X_{i}(t)=0$ if agent $i$ is susceptible and $X_{i}(t)=1$ if
it is infected. The \emph{curing process} for infected agent $i$ is a Poisson
process with curing rate $\delta\in%
\mathbb{R}
^{+}$. The \emph{infection process} for susceptible agent $i$ in contact with
infected agent $j\neq i$ is a Poisson process with infection rate $\beta\in%
\mathbb{R}
^{+}$. The competing infection processes are independent. Therefore, a
susceptible agent effectively becomes infected with rate $\beta Y_{i}(t)$,
where $Y_{i}(t)\triangleq\sum_{j=1}^{N}a_{ij}X_{j}(t)$ is the number of
infected neighbors of agent $i$ at time $t$. The ratio of the infection rate
$\beta$ over the curing $\delta$ is the \emph{infection strength}
$\tau\triangleq\frac{\beta}{\delta}$. A schematic of SIS\ epidemic spreading
model over a graph is shown in Fig. \ref{sis.eps}.%

\begin{figure}[ptb]%
\centering
\includegraphics[ width=3.5 in]{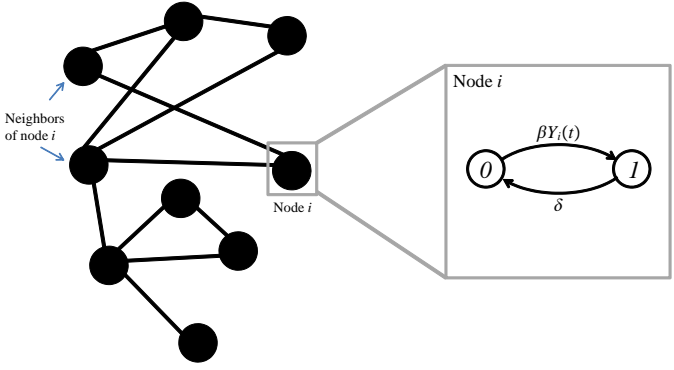}
\caption{Schematics of a contact network
along with the agent-level stochastic transition diagram for agent $i $
according to the SIS epidemic spreading model. The parameters $\beta$ and
$\delta$ denote the infection rate and curing rate, respectively. $Y_{i}(t)$
is the number of the neighbors of agent $i$ that are infected at time $t$.}%
\label{sis.eps}%
\end{figure}

Denote the infection probability of the $i$-th agent by $p_{i}\triangleq
\Pr[X_{i}(t)=1]$. It has been shown that the marginal probabilities $p_{i}$ do
not form a closed system. Actually, the exact Markov set of differential
equations has $2^{N}$ states. Van Mieghem \emph{et. al.} \cite{Piet2009TN} used
a first order mean-field type approximation to develop the NIMFA model, a set
of ordinary differential equations%
\begin{equation}
\dot{p}_{i}=\beta(1-p_{i})\sum_{j=1}^{N}a_{ij}p_{j}-\delta p_{i}%
,~i\in\{1,...,N\},\label{NInt_Model}%
\end{equation}
which represents the time evolution of the infection probability for each
agent. There is no approximation on the network topology in this model.
According to this model, it is proved that if the infection strength
$\tau=\beta/\delta$ is less than the threshold value $\tau_{c}\triangleq
\frac{1}{\lambda_{1}(A)}$, initial infection probabilities die out
exponentially. If the infection strength $\tau=\beta/\delta$ is higher than
$\tau_{c}$, then infection probabilities will go to non-zero steady state values.

\subsection{SIS Epidemic Spreading over Interconnected Networks}

Consider two groups of agents of sizes $N_{1}$ and $N_{2}$. In order to
facilitate the subsequent developments, we label the agents of the first graph
$\mathcal{G}_{1}$ from $1$ to $N_{1}$, and the agents of the second graph
$\mathcal{G}_{2}$ from $N_{1}+1$ to $N_{1}+N_{2}$. The collective adjacency
matrix $A$, defined as
\begin{equation}
A\triangleq%
\begin{bmatrix}
A_{11} & A_{12}\\
A_{21} & A_{22}%
\end{bmatrix}
\in%
\mathbb{R}
^{_{(N_{1}+N_{2})\times(N_{1}+N_{2})}},\label{A_coupled}%
\end{equation}
represents the contact between all of the agents. Since the contact topology
in this paper is undirected, $A_{11}$ and $A_{22}$ are symmetric matrices and
$A_{21}=A_{12}^{T}$. According to definition (\ref{A_coupled}), agent $i$ is
connected to agent $j$ iff $(A)_{ij}=1$. A schematic of the interconnected
contact network of the agents is represented in Fig. \ref{couplednetwork.eps}.%

\begin{figure}[ptb]%
\centering
\includegraphics[ width=3.5 in]{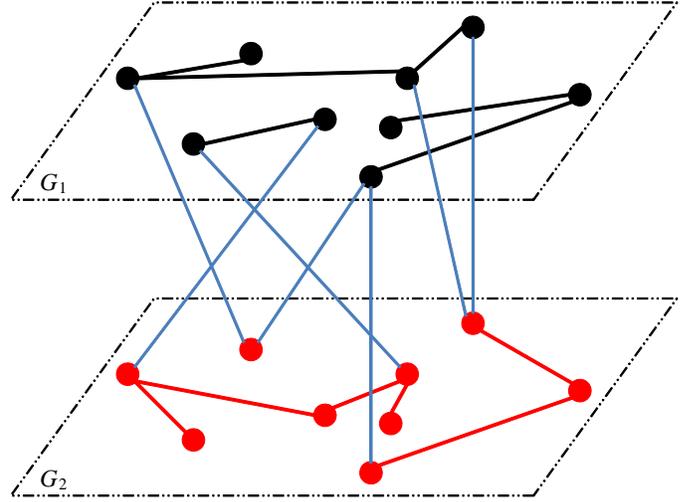}
\caption{A schematic of the coupling between
graphs $\mathcal{G}_{1}$ (in black) and $\mathcal{G}_{2}$ (in red), which each
represent a contact network. The blue links represent the coupling between the
nodes of the two graphs. $\mathcal{G}_{1}$ and $\mathcal{G}_{2}$ are not
necessarily connected. However, the whole interconnected network is
connected.}%
\label{couplednetwork.eps}%
\end{figure}

The SIS\ spreading model over a single graph described in Section
\ref{SIS_Single} can be generalized in the following way. The curing rate for
agents of graphs $\mathcal{G}_{1}$ and $\mathcal{G}_{2}$ are $\delta_{1}\in%
\mathbb{R}
^{+}$ and $\delta_{2}\in%
\mathbb{R}
^{+}$, respectively. The infection rates $\beta_{11},\beta_{12},\beta
_{21},\beta_{22}\in%
\mathbb{R}
^{+}$ are such that a susceptible agent of graph $\mathcal{G}_{m}$ receives
the infection from an infected agent in $\mathcal{G}_{n}$ with the infection
rate $\beta_{mn}$, for $m,n\in\{1,2\}$. Similar to (\ref{NInt_Model}), the
infection probabilities of the agents evolve according to the following set of
differential equations:
\begin{multline}
\dot{p}_{i}=(1-p_{i})\{\beta_{11}\sum_{j=1}^{N_{1}}a_{ij}p_{j}+\beta_{12}%
\sum_{j=N_{1}+1}^{N_{1}+N_{2}}a_{ij}p_{j}\}-\delta_{1}p_{i},\\
i\in\{1,...,N_{1}\},\label{SIS_Coupled_1}%
\end{multline}%
\begin{multline}
\dot{p}_{i}=(1-p_{i})\{\beta_{21}\sum_{j=1}^{N_{1}}a_{ij}p_{j}+\beta_{22}%
\sum_{j=N_{1}+1}^{N_{1}+N_{2}}a_{ij}p_{j}\}-\delta_{2}p_{j},\\
~i\in\{N_{1}+1,...,N_{1}+N_{2}\}.\label{SIS_Coupled_2}%
\end{multline}

\begin{remark}
Infection process is the result of interaction between a pair of agents.
Therefore, actually the infection rate in (\ref{NInt_Model}) equals to
$\beta=\mu\pi$ where $\mu\in%
\mathbb{R}
^{+}$ is the rate that an infected agents transmits the infection and $\pi
\in\lbrack0,1]$ is the probability that a susceptible agent receives a
transmitted infection. Similar arguments show that the four infection rates
$\beta_{11},\beta_{12},\beta_{21},\beta_{22}$ in (\ref{SIS_Coupled_1}) and
(\ref{SIS_Coupled_2}) are not completely independent of each other. Having
$\beta_{11}=\mu_{1}\pi_{1}$ and $\beta_{22}=\mu_{2}\pi_{2}$, the infection
rates $\beta_{12}$ and $\beta_{21}$ will have a form of $\beta_{12}=\alpha
\mu_{1}\pi_{2},\beta_{21}=\alpha\mu_{2}\pi_{1}$, where $\alpha\in%
\mathbb{R}
^{+}$ is a positive scalar accounting for heterogeneity of interconnection and
interaconnection. Therefore, the following constraint exists among the
infection rates
\end{remark}%

\begin{equation}
\beta_{11}\beta_{22}=\alpha^{2}\beta_{12}\beta_{21}.\label{beta_constraint}%
\end{equation}

Comparing (\ref{NInt_Model}) and (\ref{SIS_Coupled_1}), it can be concluded
that interconnection increases the probability of infection. This conclusion
is actually intuitive: when interconnected with other agents, there is more
possibility to receive the infection.

\section{Main Results\label{Section: Mainresult}}

\subsection{Problem Statement}

Suppose that agents of graph $\mathcal{G}_{1}$ are connected to agents of
graph $\mathcal{G}_{2}$, and the overall contact among the agents is
determined by $A$ defined in (\ref{A_coupled}), where the following assumption
holds for $\tau_{22}$.

\begin{condition}
\label{Ass: b22d2}If there is no interconnection, infection cannot survive in
$\mathcal{G}_{2}$, i.e.,%
\begin{equation}
\frac{\beta_{22}}{\delta_{2}}<\frac{1}{\lambda_{1}(A_{22})}.
\end{equation}

\end{condition}

Under Assumption \ref{Ass: b22d2}, for the infection strength $\tau_{11}=0$,
the steady state value of infection probabilities of (\ref{SIS_Coupled_1}) and
(\ref{SIS_Coupled_2}) are necessarily zero for each agent. In this paper, we
find a threshold value $\tau_{11,c}$ such that for infection strength
$\tau_{11}>\tau_{11,c}$ the steady state infection probabilities take positive
values. Since interconnection always increases the chance of receiving the
infection, we expect $\tau_{11,c}<1/\lambda_{1}(A_{11})$.

\subsection{Equation for Epidemic Threshold}

We use bifurcation theory to find the epidemic threshold. From
(\ref{SIS_Coupled_1}) and (\ref{SIS_Coupled_2}), the equilibrium points of the
infection probabilities satisfy the following set of algebraic equations%
\begin{multline}
\frac{p_{i}^{\ast}}{1-p_{i}^{\ast}}=\tau_{11}\sum_{j=1}^{N_{1}}a_{ij}%
p_{j}^{\ast}+\tau_{12}\sum_{j=N_{1}+1}^{N_{1}+N_{2}}a_{ij}p_{j}^{\ast},\\
i\in\{1,...,N_{1}\},\label{pi_ss_1}%
\end{multline}%
\begin{multline}
\frac{p_{i}^{\ast}}{1-p_{i}^{\ast}}=\tau_{21}\sum_{j=1}^{N_{1}}a_{ij}%
p_{j}^{\ast}+\tau_{22}\sum_{j=N_{1}+1}^{N_{1}+N_{2}}a_{ij}p_{j}^{\ast},\\
~i\in\{N_{1}+1,...,N_{1}+N_{2}\},\label{pi_ss_2}%
\end{multline}
where,%
\begin{equation}
\tau_{11}\triangleq\frac{\beta_{11}}{\delta_{1}},\tau_{12}\triangleq
\frac{\beta_{12}}{\delta_{1}},\tau_{21}\triangleq\frac{\beta_{21}}{\delta_{2}%
},\tau_{22}\triangleq\frac{\beta_{22}}{\delta_{2}}.\label{taws}%
\end{equation}

\begin{lemma}
\label{Lemma: SS}If the overall contact network is connected, the steady state
values of the infection probabilities are either zero for all of the agents or
absolutely positive for each agent.
\end{lemma}

\begin{proof}
The idea of the proof is inspired from \cite{Piet2009TN}. The steady state
values for the infection satisfies (\ref{pi_ss_1}) and (\ref{pi_ss_2}).
Therefore, $p_{i}^{\ast}=0$ for $\forall i\in\{1,...,N_{1}+N_{2}\}$ is a
solution for the steady state infection probabilities. Suppose there exists a
node $j$ such that $p_{j}^{\ast}>0$. According to (\ref{pi_ss_1}) and
(\ref{pi_ss_2}), for any node $i$ that is a neighbor of node $j$, i.e.,
$a_{ij}\neq0$, the steady state infection probability is%
\begin{equation}
p_{i}^{\ast}=\frac{\tau_{11}\sum_{j=1}^{N_{1}}a_{ij}p_{j}^{\ast}+\tau_{12}%
\sum_{j=N_{1}+1}^{N_{1}+N_{2}}a_{ij}p_{j}^{\ast}}{1+\tau_{11}\sum_{j=1}%
^{N_{1}}a_{ij}p_{j}^{\ast}+\tau_{12}\sum_{j=N_{1}+1}^{N_{1}+N_{2}}a_{ij}%
p_{j}^{\ast}},
\end{equation}
if $i\in\{1,...,N_{1}\}$ and%
\[
p_{i}^{\ast}=\frac{\tau_{21}\sum_{j=1}^{N_{1}}a_{ij}p_{j}^{\ast}+\tau_{22}%
\sum_{j=N_{1}+1}^{N_{1}+N_{2}}a_{ij}p_{j}^{\ast}}{1+\tau_{21}\sum_{j=1}%
^{N_{1}}a_{ij}p_{j}^{\ast}+\tau_{22}\sum_{j=N_{1}+1}^{N_{1}+N_{2}}a_{ij}%
p_{j}^{\ast}},
\]
if $i\in\{N_{1}+1,...,N_{1}+N_{2}\}$, which is positive because $\sum
_{j=1}^{N_{1}}a_{ij}p_{j}^{\ast}>0$ or $\sum_{j=N_{1}+1}^{N_{1}+N_{2}}%
a_{ij}p_{j}^{\ast}>0$. Same procedure can be applied to the neighbors of node
$i$, and so on. Hence, if the contact network is connected and at least one of
the agents have nonzero infection probability, then $p_{i}^{\ast}>0$ for all
$i\in\{1,...,N_{1}+N_{2}\}$.
\end{proof}

Before the epidemic threshold, origin is the only solution to (\ref{pi_ss_1})
and (\ref{pi_ss_2}). Epidemic threshold is the critical value $\tau_{11,c} $
such that a second equilibrium point starts leaving the origin. A corollary of
Lemma \ref{Lemma: SS} is that the epidemic threshold $\tau_{11,c}$ is such
that $p_{i}^{\ast}=0$ and $\frac{\partial p_{i}^{\ast}}{\partial\tau_{11}}>0$
for every $i\in\{1,...,N_{1}+N_{2}\}$. Taking the derivative of (\ref{pi_ss_1}%
) and (\ref{pi_ss_2}) with respect to $\tau_{11} $ at $\tau_{11}=\tau_{11,c}$
and $p_{i}^{\ast}=0$ yields
\begin{multline}
\frac{\partial p_{i}^{\ast}}{\partial\tau_{11}}=\tau_{11,c}\sum_{j=1}^{N_{1}%
}a_{ij}\frac{\partial p_{j}^{\ast}}{\partial\tau_{11}}+\tau_{12}\sum
_{j=N_{1}+1}^{N_{1}+N_{2}}a_{ij}\frac{\partial p_{j}^{\ast}}{\partial\tau
_{11}},\\
i\in\{1,...,N_{1}\},\label{dp_1}%
\end{multline}%
\begin{multline}
\frac{\partial p_{i}^{\ast}}{\partial\tau_{11}}=\tau_{21}\sum_{j=1}^{N_{1}%
}a_{ij}\frac{\partial p_{j}^{\ast}}{\partial\tau_{11}}+\tau_{22}\sum
_{j=N_{1}+1}^{N_{1}+N_{2}}a_{ij}\frac{\partial p_{j}^{\ast}}{\partial\tau
_{11}},\\
i\in\{N_{1}+1,...,N_{1}+N_{2}\}.\label{dp_2}%
\end{multline}

Defining $V_{1}\triangleq\lbrack\frac{\partial p_{1}^{\ast}}{\partial\tau
_{11}},...,\frac{\partial p_{N_{1}}^{\ast}}{\partial\tau_{11}}]^{T}$ and
$V_{2}\triangleq\lbrack\frac{\partial p_{N_{1}+1}^{\ast}}{\partial\tau_{11}%
},...,\frac{\partial p_{N_{1}+N_{2}}^{\ast}}{\partial\tau_{11}}]^{T}$, the
equations (\ref{dp_1}) and (\ref{dp_2}) can be equivalently expressed in the
collective form as%
\begin{equation}%
\begin{bmatrix}
\tau_{11,c}A_{11} & \tau_{12}A_{12}\\
\tau_{21}A_{12}^{T} & \tau_{22}A_{22}%
\end{bmatrix}%
\genfrac{[}{]}{0pt}{}{V_{1}}{V_{2}}%
=%
\genfrac{[}{]}{0pt}{}{V_{1}}{V_{2}}%
.\label{Theshold_Eq_Coupled}%
\end{equation}
The critical value of the infection strengths are those for which the above
equation has a positive solution. Equation (\ref{Theshold_Eq_Coupled}) can be
written as
\begin{align}
\tau_{11,c}A_{11}V_{1}+\tau_{12}A_{12}V_{2}  & =V_{1},\\
\tau_{21}A_{12}^{T}V_{1}+\tau_{22}A_{22}V_{2}  & =V_{2}.
\end{align}
According to Assumption \ref{Ass: b22d2}, if $V_{1}$ is positive then
$V_{2}=\tau_{21}(I-\tau_{22}A_{22})^{-1}A_{12}^{T}V_{1}$ exists and is
non-negative. Therefore, (\ref{Theshold_Eq_Coupled}) is equivalently expressed
as%
\begin{equation}
HV_{1}=V_{1}\label{HVV}%
\end{equation}
where $H$ is defined as%
\begin{equation}
H\triangleq\tau_{11,c}A_{11}+\tau_{21}\tau_{12}A_{12}(I-\tau_{22}A_{22}%
)^{-1}A_{12}^{T}.\label{H}%
\end{equation}

\subsection{Effect of Coupling on Epidemic Threshold}

The rest of the analysis is to find the threshold value $\tau_{11,c}$ such
that (\ref{HVV}) has a positive solution for $V_{1}$. The following results
facilitate the proof of Theorem \ref{Theorem: tawc11}, which is the main
result in this paper.

\begin{lemma}
\label{Paths}The number of paths of length $L$ from node $i\in\mathcal{G}_{1}$
to node $j$ corresponding to the class $(l_{1},...,l_{s})$ is:

\begin{itemize}
\item the $(i,j)$-th entry of $A_{11}^{l_{1}}A_{12}A_{12}^{l_{2}}A_{21}\cdots
A_{21}A_{1}^{l_{s}}$, if $j\in\{1,...,N_{1}\}$,

\item the $(i,j-N_{1})$-th entry of $A_{11}^{l_{1}}A_{12}A_{12}^{l_{2}}%
A_{21}\cdots A_{12}A_{2}^{l_{s}}$, if $j\in\{N_{1}+1,...,N_{1}+N_{2}\}$,
\end{itemize}

where $A_{11}^{0}=I_{N_{1}\times N_{1}}$ and $A_{22}^{0}=I_{N_{2}\times N_{2}%
}$, by convention.
\end{lemma}

\begin{proof}
We use induction for the proof. For $L=1$, the number of paths from node $i$
to $j$ is equal to $1$ if $i$ is connected to $j$, and is zero otherwise. If
$j\in\{1,...,N_{1}\}$, path of length $L=1$ corresponds to the class $(1)$.
Therefore, the number of paths from node $i$ to $j$ is equal to the $(i,j)$-th
entry of $A_{11}$. If $j\in\{N_{1}+1,...,N_{1}+N_{2}\}$, then a path of length
$L=1$ corresponds to the class either $(0,0)$. In this case, the number of
paths from node $i$ to $j$ is equal to the $(i,j-N_{1})$-th entry of
$A_{12}=A_{1}^{0}A_{12}A_{2}^{0}$. Therefore for $L=1$, the Lemma is correct.

Assume that for $L=L_{0}$ the lemma statement is correct. Consider the first
case where $j\in\{1,...,N_{1}\}$. A path of length $L=L_{0}+1$ from $i$ to $j$
is either of the class $(l_{1},...,l_{s}+1)$ or $(l_{1},...,l_{s},0)$. Such a
path can be constructed from paths of length $L_{0}$ from $i$ to $k$ of the
class $(l_{1},...,l_{s})$ then connected to node $j$ from node $k$.

If the path from $i$ to $j$ is of class $(l_{1},...,l_{s}+1)$, then the number
of such paths is%
\begin{multline*}
\sum_{k=1}^{N_{1}}(A_{11}^{l_{1}}A_{12}A_{12}^{l_{2}}A_{21}\cdots A_{21}%
A_{1}^{l_{s}})_{ik}(A_{1})_{kj}=\\
A_{11}^{l_{1}}A_{12}A_{12}^{l_{2}}A_{21}\cdots A_{21}A_{1}^{l_{s}+1}.
\end{multline*}
If the path from $i$ to $j$ is of class $(l_{1},...,l_{s},0)$, then the number
of such paths is%
\begin{multline*}
\sum_{k=N_{1}+1}^{N_{1}+N_{2}}(A_{11}^{l_{1}}A_{12}A_{12}^{l_{2}}A_{21}\cdots
A_{21}A_{1}^{l_{s}})_{i(k-N_{1})}(A_{12})_{(k-N_{1})j}\\
=A_{11}^{l_{1}}A_{12}A_{12}^{l_{2}}A_{21}\cdots A_{21}A_{1}^{l_{s}}A_{12}\\
=A_{11}^{l_{1}}A_{12}A_{12}^{l_{2}}A_{21}\cdots A_{21}A_{1}^{l_{s}}A_{12}%
A_{2}^{0}.
\end{multline*}
Hence, the theorem statement is correct for $L=L_{0}+1$ and $j\in
\{1,...,N_{1}\}$. Similar procedure can be followed to conclude the same
result for $j\in\{N_{1}+1,...,N_{1}+N_{2}\}$.
\end{proof}

\begin{theorem}
\label{The: HT}The matrix $H_{T}$ defined as%
\begin{equation}
H_{T}\triangleq A_{11}+\alpha^{2}\tau_{22}A_{12}(I-\tau_{22}A_{22})^{-1}%
A_{12}^{T}.\label{HT}%
\end{equation}
is irreducible if the coupled network is connected.
\end{theorem}

\begin{proof}
We show that%
\begin{equation}
\bar{H}_{T}\triangleq A_{11}+A_{12}A_{12}^{T}+\sum_{k=1}^{N_{2}-1}A_{12}%
A_{22}^{k}A_{12}^{T}%
\end{equation}
is irreducible. If $\bar{H}_{T}$ is shown to be irreducible, then
$A_{11}+\alpha^{2}\tau_{22}A_{12}A_{12}^{T}+\alpha^{2}\tau_{22}\sum
_{k=1}^{N_{2}-1}\tau_{22}^{k}A_{12}A_{2}^{k}A_{12}^{T}$ is irreducible. And
hence, $H_{T}=A_{11}+\alpha^{2}\tau_{22}A_{12}A_{12}^{T}+\alpha^{2}\tau
_{22}\sum_{k=1}^{\infty}\tau_{22}^{k}A_{12}A_{2}^{k}A_{12}^{T}=A_{11}%
+\alpha^{2}\tau_{22}A_{12}(I-\tau_{22}A_{22})^{-1}A_{12}^{T}$ is irreducible
and the proof is completed.

If $\mathcal{G}_{1}$ is a connected graph, then $A_{11}$ and as consequence
$\bar{H}_{T}$ is irreducible. Assume that $A_{11}$ does not represent a
connected graph. Therefore, there exists a pair $i$, $j$ such that there is no
path between them in $\mathcal{G}_{1}$. However, since the whole
interconnected network is connected, there exists a path from $i$ to $j$.
Suppose, the path is of class $(l_{1,1},l_{2,1},l_{1,2},l_{2,2},....,l_{2,s}%
,l_{1,s+1})$, i.e., it makes $l_{1,1}$ jumps in $\mathcal{G}_{1}$ to reach
vertex $k_{1}^{out}$, then it leaves $\mathcal{G}_{1}$ and enters
$\mathcal{G}_{2}$ and makes $l_{2,1}$ jumps in $\mathcal{G}_{2}$, then enters
$\mathcal{G}_{1}$ at vertex $k_{1}^{in}$. This process goes on until it makes
$l_{1,s+1}$ jumps in $\mathcal{G}_{1}$ from $k_{s}^{in}$ to reach vertex $j$.
Matrix $\bar{H}_{T}$ is proved to be irreducible if we show that $(k_{u}%
^{out},k_{u}^{in})$-th entry of $\bar{H}_{T}$ is positive for $u=1,...,s$.
Since, there is path from $k_{u}^{out}$ to $k_{u}^{in}$ which is of the class
$(0,l_{2,u},0)$, the $(k_{u}^{out},k_{u}^{in})$-th entry of $A_{12}%
A_{22}^{l_{2,u}}A_{12}^{T}\geq1$,because it is the number of such paths
according to Lemma \ref{Paths}. As a consequence, $(k_{u}^{out},k_{u}^{in}%
)$-th entry of $\bar{H}_{T}$ is positive and therefore $\bar{H}_{T}$ is
irreducible. Hence, the proof is completed.
\end{proof}

\begin{theorem}
\label{Theorem: tawc11}For
\begin{equation}
\tau_{11,c}=\dfrac{1}{\lambda_{1}(H_{T})},\label{taw11c}%
\end{equation}
where $H_{T}$ is defined in (\ref{HT}), the equation
(\ref{Theshold_Eq_Coupled}) has positive solution for $V_{1}$ and $V_{2}$.
That is $\tau_{11,c}$ is the epidemic threshold.
\end{theorem}

\begin{proof}
According to (\ref{beta_constraint}) and the definitions (\ref{taws}), we have%
\begin{equation}
\tau_{21}\tau_{12}=\alpha^{2}\tau_{11,c}\tau_{22}\text{.}%
\end{equation}
Substituting for $\tau_{21}\tau_{12}$ in (\ref{H}), equation (\ref{HVV}) gets
the form%
\begin{equation}
\tau_{11,c}H_{T}V_{1}=V_{1},\label{tawcHT}%
\end{equation}
where $H_{T}$ is defined in (\ref{HT}). In order for (\ref{tawcHT}) to have
solutions, $\tau_{11,c}$ must be the inverse of one of the eigenvalues of
$H_{T}$. However, the corresponding eigenvector $V_{1}$ must have all positive
entries. According to Theorem \ref{The: HT}, $H_{T}$ is an irreducible matrix.
Therefore, $V_{1}$ is positive only for the eigenvector corresponding to the
largest eigenvalue. Therefore, $\tau_{11,c}$ is equal to $1/\lambda_{1}%
(H_{T})$.
\end{proof}

\subsection{Interconnection Topology and Epidemic Spreading Modes}

We used a bifurcation method to find an expression for the epidemic threshold.
According to our definition, the epidemic threshold is a critical value such
that for any infection strength $\tau_{11}>\tau_{11,c}$, the steady state
values of the infection probabilities are positive. Consider the special case
where the infection strength $\tau_{22}$ is very close to the spectral radius
$\lambda_{1}(A_{22})$, i.e., $\tau_{22}\rightarrow\lambda_{1}(A_{22})$.
According to (\ref{HT}) and (\ref{taw11c}), the epidemic threshold
$\tau_{11,c}\rightarrow0$ as far as the whole contact network is connected.
This argument is true even for very weak interconnection between the two
networks $\mathcal{G}_{1}$ and $\mathcal{G}_{2}$. The reason for this
observation is that since $\tau_{22}\rightarrow\lambda_{1}(A_{22})$, only a
small amount of interconnection will lead to an outbreak in $\mathcal{G}_{2}$.
Since, the probability of infection in $\mathcal{G}_{2}$ becomes a positive
value, according to Lemma \ref{Lemma: SS}, the probability of infection in
$\mathcal{G}_{1}$ is also positive. Therefore in the case of $\tau
_{22}\rightarrow\lambda_{1}(A_{22})$ and weak interconnection, the positive
infection probabilities in $\mathcal{G}_{1}$ for small values of $\tau_{11}$
is only due to epidemic outbreak in $\mathcal{G}_{2}$.

The numerical simulations in Section \ref{Section: Simulation} illustrates
three possible curves of $\tau_{11,c}$ as a function of $\tau_{22}$, as shown
in Fig. \ref{tc1vst2.eps}. Here, the blue curve belongs to the case of weak
interconnection between the two graphs. As can be seen, the decrease in the
epidemic threshold $\tau_{11,c}$ is very slow for small values of $\tau_{22}$,
while there is a quite sharp drop in the values of $\tau_{11,c}$ as $\tau
_{22}\rightarrow\lambda_{1}(A_{22})$. For strong interconnection topology,
shown by the green curve, the value of $\tau_{11,c}$ is decreasing quickly for
small values of $\tau_{22}$. In this case, the infection in $\mathcal{G}_{1}$
starts to grow not only as the result of receiving the infection from
$\mathcal{G}_{2}$, but also as the result of a survivable internal infection
force. The red curve is an intermediate between the two spreading modes.

\begin{theorem}
The derivative $\frac{d\tau_{11,c}}{d\tau_{22}}$ at $\tau_{22}=0$ is%
\begin{equation}
\left.  \frac{d\tau_{11,c}}{d\tau_{22}}\right\vert _{\tau_{22}=0}=\frac
{\alpha^{2}\left\Vert A_{12}^{T}x_{1}\right\Vert _{2}^{2}}{\lambda_{1}%
^{2}(A_{11})},\label{DER}%
\end{equation}
where $x_{1}$ is the eigenvector of $A_{11}$ belonging to $\lambda_{1}%
(A_{11}).$
\end{theorem}

\begin{proof}
The matrix $H_{T}$ from (\ref{HT}) can be written as%
\[
H_{T}\triangleq A_{11}+\alpha^{2}\tau_{22}A_{12}A_{12}^{T}+o(\tau
_{22})\text{.}%
\]
Therefore, taking the derivative of (\ref{tawcHT}) with respect to $\tau_{22}$
at $\tau_{22}=0$ yields%
\begin{multline}
\frac{d\tau_{11,c}}{d\tau_{22}}A_{11}x_{1}+\frac{1}{\lambda_{1}(A_{11}%
)}(\alpha^{2}A_{12}A_{12}^{T})x_{1}\label{der1}\\
+(\frac{1}{\lambda_{1}(A_{11})}A_{11}-I)\frac{dV_{1}}{d\tau_{22}}=0.
\end{multline}
Multiplying (\ref{der1}) by $x_{1}^{T}$ from left, we get%
\begin{multline}
\frac{d\tau_{11,c}}{d\tau_{22}}x_{1}^{T}A_{11}x_{1}+\frac{\alpha^{2}}%
{\lambda_{1}(A_{11})}x_{1}^{T}(A_{12}A_{12}^{T})x_{1}\label{der2}\\
+x_{1}^{T}(\frac{1}{\lambda_{1}(A_{11})}A_{11}-I)\frac{dV_{1}}{d\tau_{22}}=0.
\end{multline}
In (\ref{der2}), since $x_{1}$is the normalized eigenvector of $A_{11}$
corresponding to $\lambda_{1}(A_{11})$, we have $x_{1}^{T}A_{11}x_{1}%
=\lambda_{1}(A_{11})$ and $x_{1}^{T}(\frac{1}{\lambda_{1}(A_{11})}A_{11}-I)=0$
for $A_{11}$ is symmetric. Therefore, equation (\ref{der2}) becomes%
\begin{equation}
\lambda_{1}(A_{11})\frac{d\tau_{11,c}}{d\tau_{22}}+\frac{\alpha^{2}}%
{\lambda_{1}(A_{11})}(A_{12}^{T}x_{1})^{T}(A_{12}^{T}x_{1})=0.
\end{equation}
Hence, $\frac{d\tau_{11,c}}{d\tau_{22}}$ is found to be (\ref{DER}).
\end{proof}

According to (\ref{DER}) and the proceeding arguments, we define
interconnection topology measure%
\begin{equation}
\Omega\left(  \mathcal{G}_{1},\mathcal{G}_{2}\right)  \triangleq\frac
{\alpha^{2}\left\Vert A_{12}^{T}x_{1}\right\Vert _{2}^{2}}{\lambda_{1}%
(A_{11})\lambda_{1}(A_{22})}.\label{OHM}%
\end{equation}
If $\Omega\left(  \mathcal{G}_{1},\mathcal{G}_{2}\right)  <1$, then for the
infection strength $\tau_{11}$ right above the threshold $\tau_{11,c}$ in
(\ref{taw11c}), the positive infection probability in $\mathcal{G}_{1}$ is
mostly due to external infections from $\mathcal{G}_{2}$. While if
$\Omega\left(  \mathcal{G}_{1},\mathcal{G}_{2}\right)  >1$, then for the
infection strength $\tau_{11}$ right above the threshold $\tau_{11,c}$ in
(\ref{taw11c}), the positive infection probability in $\mathcal{G}_{1}$ is
mostly due to a survivable internal infection force as the result of increased
effective level of contact among agents of $\mathcal{G}_{1}$.

\begin{remark}
A very interesting property of $\Omega\left(  \mathcal{G}_{1},\mathcal{G}%
_{2}\right)  $ defined in (\ref{OHM}) is that it is a purely topological
measure and does not depend on the epidemic parameters.
\end{remark}

\section{Numerical Simulation Results\label{Section: Simulation}}

We have generated two graphs according to the small world random network model
\cite{newman2002PNAS}. The first network has $N_{1}=500$ vertices with Watts
and Strogatz parameters for mean degree $K_{1}=10$ and the rewiring
probability $\beta_{1}=0.2$. For this graph, the spectral radius is found to
be $\lambda_{1}(A_{11})=22.0586$. The second network has $N_{2}=100$ vertices
with the Watts and Strogatz parameters for mean degree $K=2$ and the rewiring
probability $\beta_{2}=0.1$. For this graph, the spectral radius is found to
be $\lambda_{1}(A_{22})=4.3$. For the interconnection of these two graphs, we
use the following rule. All the potential edges between the first graph and
the second graph are active with some probability $\omega$, to be chosen.

In the first simulation, $\bar{\tau}_{c1}=\lambda_{1}(A_{11})\tau_{11,c}$ is
plotted as a function of $\bar{\tau}_{2}=\lambda_{1}(A_{22})\tau_{22}$, for
three different values of $\omega=0.01,0.042,0.2$. The numerical results for
the three cases are shown in Fig. \ref{tc1vst2.eps}.%
\begin{figure}[ptb]%
\centering
\includegraphics[ width=3.5 in]{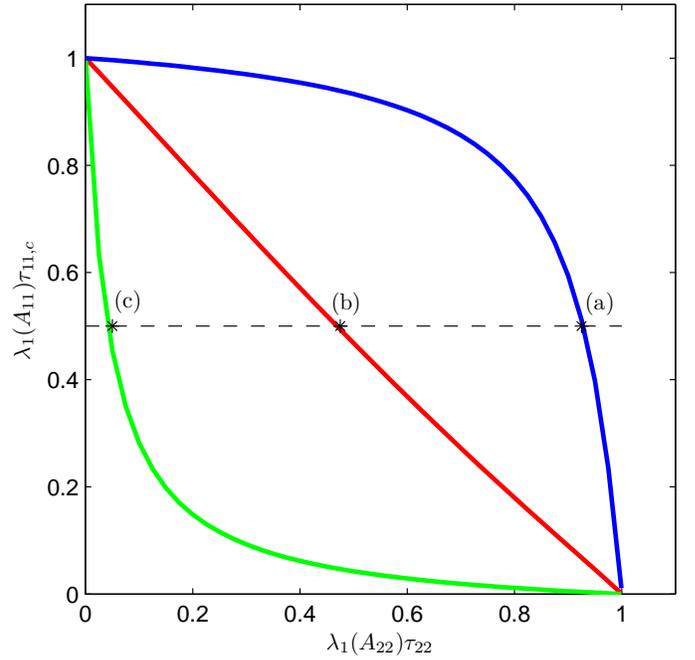}
\caption{Normalized epidemic threshold
$\bar{\tau}_{c1}=\lambda_{1}(A_{11})\tau_{11,c}$ of graph $\mathcal{G}_{1}%
$\ as a function of the normalized infection strength $\bar{\tau}_{2}%
=\lambda_{1}(A_{22})\tau_{22}$ of graph $\mathcal{G}_{2}$. The interconnection
in (a) $\omega=0.01$, the blue curve, (b) $\omega=0.042$, red curve, and (c)
$\omega=0.2,$ green curve. A $\%50$ reduction of the epidmeic threshold is
observed for the normalized infection strengths (a) $\bar{\tau}_{2}=0.925$,
(b) $\bar{\tau}_{2}=0.5$, (c) $\bar{\tau}_{2}=0.05$.}%
\label{tc1vst2.eps}%
\end{figure}
As can be seen from the blue curve, which is for $\omega=0.01$, for weak
interconnection among the graphs the emergence of positive steady state values
for the infection probability in graph $\mathcal{G}_{1}$\ is due to an
outbreak in $\mathcal{G}_{2}$. While for strong interconnection $\omega=0.2$,
which is shown with the green curve in Fig. \ref{tc1vst2.eps}, the emergence
of positive steady state values for the infection probability in graph
$\mathcal{G}_{1}$ is more due to increased level of effective contact among
the agents in $\mathcal{G}_{1}$. The red curve in Fig. \ref{tc1vst2.eps}
belongs to an intermediate interconnection strength, here $\omega=0.042$.

According to Fig. \ref{tc1vst2.eps},a $\%50$ reduction of the epidemic
threshold is observed in $\mathcal{G}_{1}$for (a) $\omega=0.01$ and $\bar
{\tau}_{2}=0.925$, (b) $\omega=0.042$ and $\bar{\tau}_{2}=0.5$, (c)
$\omega=0.2$ and $\bar{\tau}_{2}=0.05$. We have plotted the curves of $\bar
{p}_{1}^{\ast}=\frac{1}{N_{1}}\sum_{i=1}^{N_{1}}p_{i}^{\ast}$ as a function of
$\tau_{11}\lambda_{1}(A_{11})$. We have found the equilibrium values of
$p_{i}^{\ast}$ by solving the algebraic equations (\ref{pi_ss_1}) and
(\ref{pi_ss_2}). The numerical method for solving these equations is presented
in the Appendix.%

\begin{figure}[ptb]%
\centering
\includegraphics[ width=3.5 in]{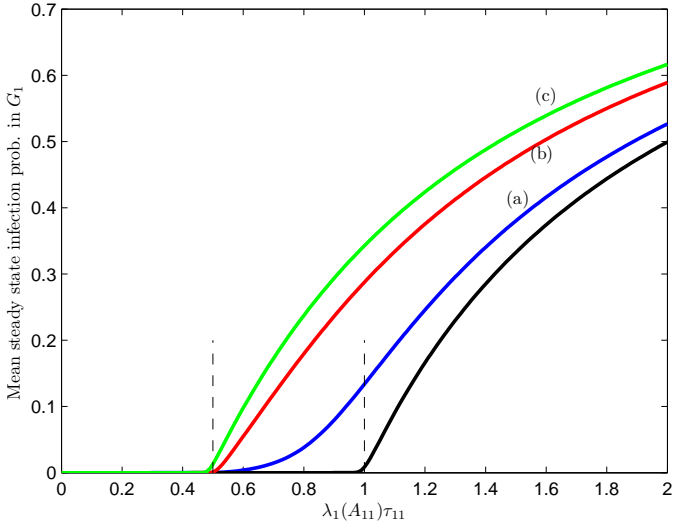}
\caption{The mean steady state infection
probability in $G_{1}$ as a function of the normalized infection strength of
$\lambda_{1}(A_{11})\tau_{11}$ for graph $G_{1} $. Black curve corresponds to
the case where there is no interconnection. In this case, the epidemic
threshold is $\tau_{11,c}=1/\lambda_{1}(A_{11})$. All the other curves
correspond to the case where $\tau_{11,c}=\frac{1}{2}\times1/\lambda
_{1}(A_{11})$. For (a) the blue curve $\omega=0.01$ and $\bar{\tau}_{2}%
=0.925$, (b) the red curve $\omega=0.042$ and $\bar{\tau}_{2}=0.5$, (c) the
green curve $\omega=0.2$ and $\bar{\tau}_{2}=0.05$.}%
\label{pss1.eps}%
\end{figure}

\section{Conclusion}

In epidemic modeling, the term infection strength indicates the ratio of
infection rate and cure rate. If the infection strength is higher than a
certain threshold then the epidemic spreads through the population and
persists in the long run. For a single generic graph representing the contact
network of the population under consideration, the epidemic threshold turns
out to be equal to the inverse of the spectral radius of the contact graph.
However, in a real world scenario it is not possible to isolate a population
completely: there is always some interconnection with another network, which
partially overlaps with the contact network. We study the spreading process of
a susceptible-infected-susceptible (SIS) type epidemic model in an
interconnected network of two general graphs. First, we prove that
interconnection always increases the probability of infection. Second, we find
that the epidemic threshold for a network interconnected to another network
with a given infection strengths rigorously derived as the spectral radius of
a new matrix which accounts for the two networks and their interconnection
links. The main contribution of this paper is the use of spectral analysis to
analyze epidemic spreading in interconnected networks. Our results have
implications for the broad field of epidemic modeling and control.

\section{Acknowledgement}

This work was supported by the National Agricultural Biosecurity Center (NABC)
at Kansas State University. It was also based on work partially supported by
the US National Science Foundation, while one of the authors, Fahmida N.
Chowdhury, was working at the Foundation. Any opinion, finding, and
conclusions or recommendations expressed in this material are those of the
authors and do not necessarily reflect the views of the National Science Foundation.

\bibliographystyle{IEEEtran}
\bibliography{InterConnectedEpidemic}

\section{Appendix}


In order to numerically solve (\ref{pi_ss_1}) and (\ref{pi_ss_2}), we can use
the recursive formula%

\begin{multline}
y_{i}(k+1)=\tau_{11}\sum_{j=1}^{N_{1}}a_{ij}\frac{y_{j}(k)}{1+y_{j}(k)}%
+\tau_{12}\sum_{j=N_{1}+1}^{N_{1}+N_{2}}a_{ij}\frac{y_{j}(k)}{1+y_{j}(k)},\\
i\in\{1,...,N_{1}\},
\end{multline}%
\begin{multline}
y_{i}(k+1)=\tau_{21}\sum_{j=1}^{N_{1}}a_{ij}\frac{y_{j}(k)}{1+y_{j}(k)}%
+\tau_{22}\sum_{j=N_{1}+1}^{N_{1}+N_{2}}a_{ij}\frac{y_{j}(k)}{1+y_{j}(k)},\\
i\in\{N_{1}+1,...,N_{1}+N_{2}\},
\end{multline}
which converges to the fixed points $y_{i}^{\ast}:=\frac{p_{i}^{\ast}}%
{1-p_{i}^{\ast}}$, for $i\in\{1,...,N_{1}+N_{2}\}$. The steady state values of
the infection probabilities are then $p_{i}^{\ast}=\frac{y_{i}^{\ast}}%
{1+y_{i}^{\ast}}$.

\end{document}